\documentclass[reqno]{amsart}
\usepackage[utf8]{inputenc}
\usepackage{amsmath}
\usepackage{amssymb}
\usepackage{xcolor}
\usepackage{graphicx}
\usepackage{float}
\usepackage[hidelinks]{hyperref}
\usepackage{cite}
\newtheorem{theorem}{Theorem}
\newtheorem{lemma}{Lemma}
\newtheorem{definition}{Definition}

\newcommand{\C}{\mathbb{C}}
\newcommand{\R}{\mathbb{R}}

\newcommand{\ler}[1]{\left( #1 \right)}
\newcommand{\lesq}[1]{\left[ #1 \right]}
\newcommand{\lecurl}[1]{\left\{ #1 \right\}}
\newcommand{\hohc}{\cH \otimes \cH^*}
\newcommand{\Sig}{\mathbf{\sigma}}
\newcommand{\abs}[1]{\left| #1 \right|}

\newcommand{\inner}[2]{\left< #1,#2 \right>}
\newcommand{\norm}[1]{\left|\left|#1\right|\right|}
\newcommand{\Bra}[1]{\langle\langle #1 \right| \right|}
\newcommand{\Ket}[1]{\left| \left| #1 \rangle\rangle}
\newcommand{\isom}[1]{\mathrm{Isom} \ler{#1}}

\newcommand{\be}{\begin{equation}}
\newcommand{\ee}{\end{equation}}
\newcommand{\ba}{\begin{array}}
\newcommand{\ea}{\end{array}}

\newcommand{\fel}{\frac{1}{2}}
\newcommand{\cH}{\mathcal{H}}
\newcommand{\cK}{\mathcal{K}}
\newcommand{\cP}{\mathcal{P}}
\newcommand{\cD}{\mathcal{D}}
\newcommand{\cF}{\mathcal{F}}
\newcommand{\cG}{\mathcal{G}}
\newcommand{\cS}{\mathcal{S}}
\newcommand{\cC}{\mathcal{C}}
\newcommand{\cU}{\mathcal{U}}
\newcommand{\cW}{\mathcal{W}}
\newcommand{\cL}{\mathcal{L}}
\newcommand{\bb}{\mathbf{b}}
\newcommand{\bB}{\mathbf{B}}
\newcommand{\bO}{\mathbf{O}}
\newcommand{\bS}{\mathbf{S}}
\newcommand{\tr}{\mathrm{tr}}
\newcommand{\sut}{\mathbf{SU}(2)}

\title{Quantum Wasserstein isometries on the qubit state space}

\author[Gy\"orgy P\'al Geh\'er]{Gy\"orgy P\'al Geh\'er}
\address{Gy\"orgy P\'al Geh\'er, Riverlane, 59 St Andrew’s St, Cambridge CB2 3DD, United Kingdom
}
\email{gehergyuri@gmail.com}

\author[J\'ozsef Pitrik]{J\'ozsef Pitrik}
\address{J\'ozsef Pitrik, Wigner Research Centre for Physics\\ Budapest H-1525, Hungary\\ and Alfr\'ed R\'enyi Institute of Mathematics\\ Re\'altanoda u. 13-15.\\ Budapest H-1053\\ Hungary\\ and Department of Analysis, Institute of Mathematics \\Budapest University of Technology and Economics\\ M\H{u}egyetem rkp. 3. \\ Budapest H-1111\\ Hungary}
\email{pitrik.jozsef@renyi.hu}

\author[Tam\'as Titkos]{Tam\'as Titkos}
\address{Tam\'as Titkos, Alfr\'ed R\'enyi Institute of Mathematics\\ Re\'altanoda u. 13-15.\\
Budapest H-1053\\ Hungary\\ and BBS University of Applied Sciences\\ Alkotm\'any u. 9.\\
Budapest H-1054\\ Hungary}

\email{titkos.tamas@renyi.hu}

\author[D\'aniel Virosztek]{D\'aniel Virosztek}
\address{D\'aniel Virosztek, Alfr\'ed R\'enyi Institute of Mathematics\\ Re\'altanoda u. 13-15.\\Budapest H-1053\\ Hungary}

\email{virosztek.daniel@renyi.hu}

\date{}



\thanks{Geh\'er was supported by the Leverhulme Trust Early Career Fellowship (ECF-2018-125), and also by the Hungarian National Research, Development and Innovation
Office (Grant no. K115383); Pitrik was supported by the “Frontline” Research Excellence Programme of the NKFIH
(Grant No. KKP133827); Titkos was supported by the Hungarian National Research, Development and Innovation Office - NKFIH (grant no. K115383); Virosztek was supported by the Momentum Program of the Hungarian Academy of Sciences (grant no. LP2021-15/2021)
and partially supported by the Hungarian National Research, Development and Innovation Office – NKFIH (grants no. K124152 and KH129601).}

\begin{document}

\maketitle

\begin{abstract}
We describe Wasserstein isometries of the quantum bit state space with respect to distinguished cost operators.
We derive a Wigner-type result for the cost operator involving all the Pauli matrices: in this case, the isometry group consists of unitary or anti-unitary conjugations. In the Bloch sphere model this means that the isometry group coincides with the classical symmetry group $\bO(3).$ On the other hand, for the cost generated by the qubit \lq\lq clock" and \lq\lq shift" operators, we discovered non-surjective and non-injective isometries as well, beyond the regular ones. This phenomenon mirrors certain surprising properties of the quantum Wasserstein distance.
\end{abstract}

\section{Introduction}

The relevance of quadratic Wasserstein spaces has grown dramatically in recent years due to their close connection with the theory of optimal transportation. Recall that if $(X,r)$ is complete and separable metric space, then the classical quadratic Wasserstein space $\mathcal{W}_2(X)$ is the collection of those probability measures on the Borel $\sigma$-algebra $\mathcal{B}$ that satisfy
$\int_X r(x,x_0)^2~\mathrm{d}\mu(x)<\infty$ for some $x_0\in X$,
endowed with the metric
\begin{equation*}
d_2(\mu, \nu):=\left(\inf_{\pi} \int_{X \times X} r^2(x,y)~\mathrm{d} \pi(x,y)\right)^{1/2}.
\end{equation*}
The infimum above is taken over the set of all couplings (or transport plans), that is, the set of all probability measures on $X \times X$ whose first marginal is $\mu$ and the second marginal is $\nu$. The Wasserstein distance quantifies the minimal effort required to morph $\mu$ into $\nu$ when the cost of transporting a unit mass from $x$ to $y$ is $r^2(x,y)$. Methods based on the theory of optimal transport and nice properties of Wasserstein spaces have achieved great success in several important fields of pure mathematics including probability theory \cite{bgl,Butkovsky}, theory of (stochastic) partial differential equations \cite{hairer,navier-stokes}, variational problems \cite{Figalli1,Figalli2} and geometry of metric spaces \cite{LV,VRS,Sturm}. 
Besides theoretical applications, the geometric characteristics of the Wasserstein metric (and other transport-related metrics) have given a momentum for research in many areas of applied sciences like image processing \cite{imageprocessing1,imageprocessing2}, medical imaging \cite{medicalimaging1,medicalimaging2}, inverse imaging problems \cite{inv}, and machine learning \cite{m1,MachineLearning2,PC,MachineLearning1,MachineLearning3,m2}.

It is a general phenomenon that concepts and notions that are well-established in the classical commutative world do not have a unique \lq\lq best" extension in the non-commutative world but there are many possible ways of generalization with pros and cons. This is the case concerning optimal transportation as well.
Non-commutative optimal transport is a flourishing research field these days with several different promising approaches such as that of Biane and Voiculescu \cite{Bi-Vo}, Carlen, Maas, Datta, and Rouzé  \cite{CM,CM2,CM3,DR1,DR2}, Golse, Mouhot, and Paul \cite{CGP, CGP2, Golse-Mouhot-Paul, Golse-Paul, Golse1,Golse2}, De Palma and Trevisan \cite{DP-M-T-L,DPT},  \.Zyczkowski and his collaborators  \cite{KZ1,KZ2,KZ3,Zs98, Zsi01}, and Duvenhage \cite{RD1,RD2}.
From our viewpoint, the most relevant approach is the one of De Palma and Trevisan involving quantum channels. We aim to explore the structure of isometries with respect to two distinguished quantum optimal transport distances. By isometry we mean a self-map of the state space preserving the quantum Wasserstein distance, without any further assumptions on surjectivity or injectivity. The existence of non-injective isometries is a surprising and fascinating phenomenon that cannot occur in a genuine metric setting. However, according to many of the approaches including that of De Palma and Trevisan \cite{DPT} which we follow, the quantum Wasserstein distance of states \emph{is not a genuine metric,} e.g., states may have positive distance from themselves.

When working with a structure which carries a distance-structure, a natural question arises: \emph{can we describe the structure of distance preserving maps?} As Hermann Weyl said in \cite{Weyl}: “Whenever you have to do with a structure–endowed entity $\Sigma$ try to determine its group of automorphisms, the group of those element–wise transformations which leave all structural relations undisturbed. You can expect to gain a deep insight into the constitution of $\Sigma$ in this way.” In recent years, there has been a lot of activity concerning such questions, see e.g. \cite{BZ,bk2016,DolinarKuzma,DolinarMolnar,FJ1,FJ2,Kuiper,LP,JMAA,TAMS, RIMS,HIL,TnSn,Kloeckner-2010,Levy,Monclair,Niemiec,S-R,titkoskissgraf,Virosztek}. We highlight three papers which deals with the structure of Wasserstein isometries over Euclidean spaces \cite{TAMS,HIL,Kloeckner-2010}. In \cite{Kloeckner-2010} Kloeckner described the isometry group of the quadratic Wasserstein space $\mathcal{W}_2(\mathbb{R}^n)$.
Later in \cite{TAMS} and \cite{HIL} we gave a complete characterisation of isometries of
$p$-Wasserstein spaces over real and separable Hilbert spaces for all parameters $1\leq p < \infty$.

In this paper, we consider the quantum case, namely we study Wasserstein isometries of the quantum bit state space with respect to two distinguished cost operators. In Theorem \ref{thm:wigner-symmetry} we derive a Wigner-type result for the cost operator involving all the Pauli matrices: in this case, the isometry group consists of unitary or anti-unitary conjugations. In the Bloch sphere model this means that the isometry group coincides with the classical symmetry group $\bO(3).$ In Theorem \ref{thm:clock-shift} we provide a lower and an upper bound for the isometry semigroup in the case when the cost is governed by the \lq\lq clock" and \lq\lq shift" operators. In order to determine the actual isometry semigroup, we performed some numerical test, see Section \ref{s: numerics}, which suggests that the isometry semigroup coincides with the lower bound in Theorem \ref{thm:clock-shift}. To the best of our knowledge, this work is the first one concerning quantum Wasserstein isometries.

\section{Basic notions, notation}
Throughout this paper $\cH$ will denote the Hilbert space $\C^2.$ The symbol $\mathcal{L}(\cH)$ stands for the set of all linear operators on $\cH$, and $A\leq B$ means that $B-A$ is positive semidefinite. The set of quantum states will be denoted by $\mathcal{S}(\cH)$, that is, $\mathcal{S}(\cH)=\{\rho\in\mathcal{L}(H)\,|\,\rho\geq0,\,\tr_{\cH}\rho=1\}$. A state $\rho\in\mathcal{S}(\cH)$ is called a \emph{pure state} if it is a rank one projection, i.e., there exist a unit vector $\psi\in\cH$ such that $\rho=|\psi\rangle\langle\psi|$. The set of pure states will be denoted by $\cP_1\ler{\cH}$. For a $\Pi\in\mathcal{L}(\cH \otimes \cH^*)$ the partial trace $\tr_{\cH^*}\Pi$ is defined by $\tr_{\cH}((\tr_{\cH^*}\Pi)A)=\tr_{\cH\otimes \cH^*}(\Pi(A\otimes I_{\cH^*}))$ for all $A\in\mathcal{L(\cH)}$, and similarly, the partial trace $\tr_{\cH}\Pi$ is defined by $\tr_{\cH^*}((\tr_{\cH}\Pi)B^T)=\tr_{\cH\otimes\cH^*}(\Pi(I_{\cH}\otimes B^T))$ for all $B\in\mathcal{L}(\cH)$. Using the canonical linear isomorphism between $\mathcal{L}(\cH)$ and $\cH\otimes\cH^*$, for an operator $A\in\mathcal{L}(\cH)$ the symbol $|| A \rangle\rangle$ denotes the corresponding vector in $\cH\otimes\cH^*$.

Elements of $\mathcal{S}(\cH)$ can be represented by vectors using the Bloch representation. The \emph{Bloch vector} $\bb_{\rho}$ of a state $\rho \in \mathcal{S}(\cH)$ is defined by
$$
\R^3 \ni \bb_{\rho}:=\ler{\tr_{\cH}\ler{\rho \sigma_j} }_{j=1}^3
$$
where the $\sigma_j$'s are the Pauli operators
\be \label{eq:pauli}
\sigma_1=\lesq{\ba{cc} 0 & 1  \\  1 & 0  \ea}\qquad \sigma_2=\lesq{\ba{cc} 0 & -i  \\  i & 0  \ea}\qquad \sigma_3=\lesq{\ba{cc} 1 & 0  \\  0 & -1  \ea}.
\ee
The positivity condition $\rho \geq 0$ ensures that $\norm{\bb_{\rho}}_{\R^3} \leq 1,$ and hence we will denote the Bloch ball by $\mathbf{B}^3$. The symbols $\mathbf{U}(n)$, $\mathbf{SU}(n)$, $\mathbf{O}(n)$, and $\mathbf{SO}(n)$ denote the unitary, special unitary, orthogonal, and special orthogonal groups, respectively.
Following the convention of De Palma and Trevisan \cite{DPT}, the set of all couplings of the quantum states $\rho, \omega \in \cS\ler{\cH}$ is denoted by $\cC\ler{\rho, \omega},$ and is given by

\be \label{eq:q-coup-def}
\cC\ler{\rho, \omega}=\lecurl{\Pi \in \cS\ler{\cH \otimes \cH^*} \, \middle| \, \tr_{\cH^*} \Pi=\omega, \,  \tr_{\cH} \Pi=\rho^T}.
\ee
We remark that $\cC\ler{\rho,\omega}$ is never empty, because the trivial coupling $\omega\otimes\rho^T$ belongs to $\cC\ler{\rho,\omega}$.
According to the convention \eqref{eq:q-coup-def}, we consider \emph{quadratic cost operator}s of the form
\be \label{eq:cost-op-def}
C=\sum_{j=1}^K \ler{A_j \otimes I_{\cH^*} - I_{\cH} \otimes A_j^T}^2,
\ee
where the $A_j$'s are self-adjoint operators on $\cH.$ The reason of this choice is that the $A_j$'s represent observable physical quantities and for a state $\Pi$ of the composite system, the quantity $\tr_{\hohc}(\Pi\, C)$ is the expected quadratic difference between outcomes on the first and second subsystems.
\par
The corresponding \emph{quadratic quantum Wasserstein distance} $D_C\ler{\rho, \omega}$ of $\rho$ and $\omega$ is defined by
\be \label{eq:qw-dist-def}
D_C\ler{\rho, \omega}=\left(\inf_{\Pi \in \cC\ler{\rho, \omega}} \tr_{\hohc} \ler{\Pi \, C}\right)^{1/2}.
\ee
The main goal of this paper is to describe the structure of isometries, that is, quantum Wasserstein distance preserving maps, of the qubit state space. As the quantum version of the Wasserstein distance is not a genuine metric, e.g., states may have a positive distance from themselves, we precisely state below what we mean by isometry. 
\begin{definition}[Quantum Wasserstein isometry] \label{def:QW-isom}
A map $\Phi: \, \cS\ler{\cH} \rightarrow \cS\ler{\cH}$ is called a \emph{quantum Wasserstein isometry} with respect to the cost operator $C$ if
\be \label{eq:QW-isom}
D_C\ler{\Phi\ler{\rho}, \Phi\ler{\omega}}=D_C\ler{\rho, \omega}
\ee
for all $\rho, \omega \in \cS\ler{\cH}.$
\end{definition}

Note that in Definition \ref{def:QW-isom} there is no a priori assumption on the surjectivity or injectivity of $\Phi.$ In fact, a \emph{typical} quantum Wasserstein isometry described in Theorem \ref{thm:clock-shift} is neither surjective nor injective! See \eqref{eq:weird-map} for very simple examples of isometries of this kind. The fact that states may have a positive distance from themselves opens the door for non-injective isometries which cannot exist in a genuine metric setting.

\section{ Symmetric cost operator: a Wigner-type result}\label{sec:Wtype1}

Let us consider the cost operator which is symmetric in the sense that it involves all the Pauli operators $\sigma_1,\sigma_2,\sigma_3$ --- see \eqref{eq:pauli}. The symmetric cost is defined by
\be \label{eq:symm-cost-def}
C_{sym}:=\sum_{j=1}^3 \ler{\sigma_j \otimes I_{\cH^*} - I_{\cH} \otimes \sigma_j^T}^2
=\lesq{\ba{cccc} 4 & 0 & 0 & -4 \\ 0 & 8 & 0 & 0 \\ 0 & 0 & 8 & 0 \\ -4 & 0 & 0 & 4 \ea},
\ee
and the corresponding quantum Wasserstein distance is denoted by $D_{sym}\ler{\cdot,\cdot}.$
An important feature of the symmetric cost is that the induced distance $D_{sym}$ is invariant under unitary or anti-unitary conjugations (in short: Wigner symmetries).

\begin{lemma} \label{lem:wig-sym}
For any unitary or anti-unitary operator $U$ and for any $\rho, \omega \in \cS\ler{\cH}$ we have
\be\label{eq:unit-sim-isom}
D_{sym}\ler{U \rho U^*, U \omega U^*}=D_{sym}\ler{\rho, \omega}.
\ee
\end{lemma}

\begin{proof}
We note first that in view of \eqref{eq:q-coup-def}, for any $U \in \mathbf{U}(2)$ and $\rho,\omega\in\cS\ler{\cH}$ we have
\be \label{eq:coup-unit-inv}
\cC\ler{U \rho U^*, U \omega U^*}=\lecurl{\cU \, \Pi \,\cU^* \, \middle| \, \Pi \in \cC \ler{\rho, \omega} }
\ee
where
\be \label{eq:cU-def}
\cU:=U \otimes \ler{U^T}^*
\ee
is a unitary on $\cH \otimes \cH^*.$
\par
The relation $\lecurl{\cU \, \Pi \,\cU^* \, \middle| \, \Pi \in \cC \ler{\rho, \omega} } \subseteq \cC\ler{U \rho U^*, U \omega U^*}$ is justified as follows.
Let $\Pi=\sum_{r=1}^R A_r \otimes B_r^T \, \ler{A_r,B_r \in \cL\ler{\cH}}$ be a decomposition of a coupling $\Pi \in \cC \ler{\rho, \omega}.$
Then
$$
\cU \, \Pi \, \cU^*=\ler{U \otimes \ler{U^T}^*}\ler{\sum_{r=1}^R A_r \otimes B_r^T}\ler{U^* \otimes U^T}
=\sum_{r=1}^R U A_r U^* \otimes \ler{U B_r U^*}^T
$$
and hence
$$
\tr_{\cH^*} \ler{\cU \, \Pi \, \cU^*}= U \ler{\sum_{r=1}^R \ler{\tr_{\cH^*} B_r^T} A_r} U^* =U \omega U^*
$$
and very similarly, $\tr_{\cH} \ler{\cU \, \Pi \,\cU^*}=\ler{U \rho U^*}^T,$ and unitary conjugation preserves positivity. On the other hand, if $\Sigma \in \cC\ler{U \rho U^*, U \omega U^*},$ then $\cU^* \Sigma \cU \in \cC\ler{\rho, \omega},$ where $\cU^*=U^* \otimes U^T,$ so the equation \eqref{eq:coup-unit-inv} is justified.   
\par
The next step is to show that the symmetric cost defined by \eqref{eq:symm-cost-def} is \emph{invariant} under conjugation by $\cU,$ more precisely,
\be \label{eq:cost-invariance}
\ler{U \otimes \ler{U^T}^*} C_{sym} \ler{U^* \otimes U^T}=C_{sym},
\ee
for any $U \in \mathbf{U}(2).$ As $\ler{\gamma U} \otimes \ler{\ler{\gamma U}^T}^*=U \otimes \ler{U^T}^*$ for any complex number $\gamma$ of modulus one, it is sufficient to prove \eqref{eq:cost-invariance} for $U \in \sut.$
\par
Let $U=\lesq{\ba{cc} \alpha & -\overline\beta \\ \beta & \overline\alpha \ea} \in \sut,$ that is, $\alpha, \beta \in \C$ and $\abs{\alpha}^2+\abs{\beta}^2=1.$ Then
\be \label{eq:cU-form}
U \otimes \ler{U^T}^*
=\lesq{\ba{cccc}
\abs{\alpha}^2 & -\alpha \beta & -\overline{\alpha \beta} & \abs{\beta}^2\\
\alpha \overline{\beta} & \alpha^2 & -\overline{\beta}^2 & -\alpha \overline{\beta} \\
\beta \overline{\alpha} & -\beta^2 & \overline{\alpha}^2 & -\beta \overline{\alpha} \\
\abs{\beta}^2 & \alpha \beta & \overline{\alpha \beta} & \abs{\alpha}^2 \ea}.
\ee
The spectral decomposition of the cost operator is
$$
C_{sym}= \lesq{\ba{cccc} 4 & 0 & 0 & -4 \\ 0 & 8 & 0 & 0 \\ 0 & 0 & 8 & 0 \\ -4 & 0 & 0 & 4 \ea}
=0 \cdot \lesq{\ba{cccc} 1 & 0 & 0 & 1 \\ 0 & 0 & 0 & 0 \\ 0 & 0 & 0 & 0 \\ 1 & 0 & 0 & 1 \ea}
+
$$
$$
+4 \cdot \ler{\lesq{\ba{cccc} 1 & 0 & 0 & -1 \\ 0 & 0 & 0 & 0 \\ 0 & 0 & 0 & 0 \\ -1 & 0 & 0 & 1 \ea}
+\lesq{\ba{cccc} 0 & 0 & 0 & 0 \\ 0 & 1 & 1 & 0 \\ 0 & 1 & 1 & 0 \\ 0 & 0 & 0 & 0 \ea}+\lesq{\ba{cccc} 0 & 0 & 0 & 0 \\ 0 & 1 & -1 & 0 \\ 0 & -1 & 1 & 0 \\ 0 & 0 & 0 & 0 \ea}},
$$
that is, 
$$
C_{sym}=0 \cdot \Ket{I_\cH}\Bra{I_\cH} + 4 \cdot \ler{\Ket{\sigma_3}\Bra{\sigma_3}+\Ket{\sigma_1}\Bra{\sigma_1}+\Ket{\sigma_2}\Bra{\sigma_2}}.
$$
Note that $\Ket{I_\cH}\Bra{I_\cH}$ is an eigenprojection of the normal operator $U \otimes \ler{U^T}^*,$ see \eqref{eq:cU-form}, and hence $U \otimes \ler{U^T}^*$ commutes with the cost $C_{sym},$ which is a rank three projection up to a multiplicative constant.
So we deduced that $\left[C_{sym}\, , \, U \otimes \ler{U^T}^*\right]=0$ for any $U \in \sut,$ which is equivalent to the desired equation \eqref{eq:cost-invariance}.
Therefore,
$$
D_{sym}^2\ler{U \rho U^*, U \omega U^*}=\inf \lecurl{\tr_{\hohc} \ler{C_{sym} \Sigma} \, \middle| \, \Sigma \in \cC\ler{U \rho U^*, U \omega U^*}}
$$
$$
=\inf \lecurl{\tr_{\hohc} \ler{C_{sym} \cU \, \Pi \, \cU^*} \, \middle| \, \Pi \in \cC\ler{\rho, \omega}}
$$
$$
=\inf \lecurl{\tr_{\hohc} \ler{ \cU^* \, C_{sym} \cU \, \Pi } \, \middle| \, \Pi \in \cC\ler{\rho, \omega}}
$$
\be \label{eq:d-sym-inv}
=\inf \lecurl{\tr_{\hohc} \ler{C_{sym} \, \Pi } \, \middle| \, \Pi \in \cC\ler{\rho, \omega}}=D_{sym}^2\ler{\rho,  \omega},
\ee
so \eqref{eq:unit-sim-isom} is proved for unitaries.
\par
As for an anti-unitary $U,$ let us note that $U A U^*=V \overline{A} V^{*}$ for some unitary $V \in \mathbf{U}(2).$ However, $\overline{\sigma_j} \otimes \overline{\sigma_j}=\sigma_j \otimes \sigma_j$ for every $j \in \lecurl{1,2,3},$ hence $\overline{C_{sym}}=C_{sym},$ and
$\cC\ler{\overline{\rho} ,\overline{\omega}}=\lecurl{\overline{\Pi} \, \middle| \, \Pi \in \cC \ler{\rho, \omega}}.$
Therefore,
$$
D_{sym}^2\ler{\overline{\rho} ,\overline{\omega}}=\inf \lecurl{\tr_{\hohc} \ler{\overline{C_{sym}} \, \overline{\Pi}} \, \middle| \, \Pi \in \cC\ler{\rho, \omega}}=D_{sym}^2\ler{\rho ,\omega}
$$
for any $\rho, \omega \in \cS\ler{\cH},$ and \eqref{eq:unit-sim-isom} is proved for anti-unitaries.
\end{proof}

Now we know that Wigner symmetries are quantum Wasserstein isometries with respect to the distance induced by the symmetric transport cost operator \eqref{eq:symm-cost-def}. The main result of this section is that \emph{all the quantum Wasserstein isometries} with respect to this cost are Wigner symmetries.

\begin{theorem} \label{thm:wigner-symmetry}
Let $\Phi: \, \cS\ler{\cH} \rightarrow \cS\ler{\cH}$ be a quantum Wasserstein isometry with respect to the cost operator $C_{sym}$ given in \eqref{eq:symm-cost-def}. That is, assume that
$$
D_{sym}\ler{\Phi\ler{\rho}, \Phi\ler{\omega}}=D_{sym}\ler{\rho, \omega} \qquad \ler{\rho, \omega \in \cS\ler{\cH}}.
$$
Then there exist a unitary or anti-unitary operator $U$ acting on $\cH=\C^2$ such that
\be \label{eq:qw-pres-form}
\Phi\ler{\rho}=U \rho U^* \qquad \ler{\rho \in \cS\ler{\cH}}.
\ee
Conversely, any map of the form \eqref{eq:qw-pres-form} is a quantum Wasserstein isometry with respect to $C_{sym}$.
\par
In other words, the isometry group of the quantum Wasserstein space defined by the cost operator $C_{sym}$ coincides with the orthogonal group $\mathbf{O}(3)$ by the Bloch representation. 
\end{theorem}

\begin{proof}
Let 
\be \label{eq:rho-def}
\rho=\fel\ler{I_{\cH}+x \sigma_1+y \sigma_2 + z \sigma_3}=\fel \lesq{\ba{cc} 1+z & x - y i \\ x + y i & 1-z \ea}
\ee
and
\be \label{eq:omega-def}
\omega=\fel\ler{I_{\cH}+u \sigma_1+v \sigma_2 + w \sigma_3}=\fel \lesq{\ba{cc} 1+w & u - v i \\ u + v i & 1-w \ea}
\ee
be states on $\cH.$
The cost of the trivial coupling of $\rho$ and $\omega$ is
$$
\tr_{\hohc} \ler{\ler{\omega \otimes \rho^T} C_{sym}}=\tr_{\hohc} \ler{\ler{\omega \otimes \rho^T} \ler{6 I_{\hohc} -2 \sum_{j=1}^3 \sigma_j \otimes \sigma_j^T}}
$$
\be \label{eq:triv-coupl-cost}
=6-2 \sum_{j=1}^3 \tr_{\cH} \ler{\omega \sigma_j} \cdot \tr_{\cH^*} \ler{\rho^T \sigma_j^T} = 6- 2 \ler{xu+yv+zw}=6-2 \inner{\bb_\rho}{\bb_\omega}
\ee
where $\bb_\rho=\ler{x,y,z} \in \bB^3$ and $\bb_\omega=\ler{u,v,w} \in \bB^3$ are the Bloch vectors of $\rho$ and $\omega,$ respectively.
\par
We see from \eqref{eq:triv-coupl-cost} that $D_{sym}^2\ler{\rho, \omega} \leq 8$ for any $\rho$ and $\omega,$ and $D_{sym}^2\ler{\rho, \omega}=8$ if and only if $\bb_\omega=-\bb_\rho$ and $\norm{\bb_\rho}=1.$ This latter property amounts to $\rho$ and $\omega$ being orthogonal pure states. Therefore,
\be \label{eq:max-dist}
\max_{\rho, \omega \in \cS\ler{\cH}} D_{sym}^2\ler{\rho, \omega}=8,
\ee
and the maximum is attained if and only if $\rho$ and $\omega$ are orthogonal pure states.
Consequently, \emph{any quantum Wasserstein isometry with respect to $C_{sym}$ maps pure states to pure states.}
\par
Now, we exploit the fact that if either $\rho$ or $\omega$ is a pure state, then the trivial coupling is the only quantum coupling, that is, $\cC\ler{\rho, \omega}=\lecurl{\omega \otimes \rho^T}.$ This is the quantum analogue of the classical fact that if either $\mu$ or $\nu$ is a Dirac mass, then the only classical coupling of $\mu$ and $\nu$ is $\mu \times \nu.$ So, the quantum Wasserstein distance of pure states $\rho, \omega \in \cP_1\ler{\cH}$ can be expressed in terms of their Bloch vectors as follows:
\be \label{eq:pure-st-dist}
D_{sym}^2\ler{\rho, \omega}=6-2\inner{\bb_\rho}{\bb_\omega}=4+\norm{\bb_\rho-\bb_\omega}^2 \qquad \ler{\rho, \omega \in \cP_1\ler{\cH}}.
\ee
This means that the squared Wasserstein distance of pure states is an affine image of the Euclidean distance of the corresponding Bloch vectors. In the Bloch representation, pure states correspond to the unit sphere $\bS^2 \subset \R^3.$ Therefore, if $\Phi: \, \cS\ler{\cH} \rightarrow \cS\ler{\cH}$ is a Wasserstein isometry for $C_{sym},$ then it acts on pure states like an isometry of $\bS^2.$ Namely, there exists an $O \in \bO(3)$ such that 
\be \label{eq:isom-O3}
\bb_{\Phi\ler{\rho}}=O\ler{\bb_\rho} \qquad \ler{\rho \in \cP_1\ler{\cH}}.
\ee
Observe now that if a Wasserstein isometry $\Phi: \, \cS\ler{\cH} \rightarrow \cS\ler{\cH}$ leaves the distinguished pure states $\omega_j:=\fel\ler{I_{\cH}+\sigma_j}, \, j \in \lecurl{1,2,3}$ invariant, then
$$
\Phi=\mathrm{Id}_{\cS\ler{\cH}}.
$$
Indeed, if $\rho \in \cS\ler{\cH}$ is given by \eqref{eq:rho-def}, then by \eqref{eq:triv-coupl-cost}, we have 
\be \label{eq:coordinates}
D_{sym}^2\ler{\rho, \omega_1}=6-2x, \, D_{sym}^2\ler{\rho, \omega_2}=6-2y, \, D_{sym}^2\ler{\rho, \omega_3}=6-2z.
\ee
Therefore, if $\Phi\ler{\omega_j}=\omega_j$ for every $j \in \lecurl{1,2,3},$ then the preserver equation \eqref{eq:QW-isom} tells us that all Bloch coordinates of $\rho$ remain the same, that is, $\Phi\ler{\rho}=\rho.$
Consequently, any Wasserstein isometry $\Phi: \, \cS\ler{\cH} \rightarrow \cS\ler{\cH}$ acts as an isometry of $\bB^3,$ namely, there exists an $O \in \bO(3)$ such that 
\be \label{eq:isom-full-O3}
\bb_{\Phi\ler{\rho}}=O\ler{\bb_\rho} \qquad \ler{\rho \in \cS\ler{\cH}}.
\ee
It is clear that unitary conjugations on $\cS\ler{\cH}$ induce orientation preserving orthogonal transformations on the Bloch sphere $\bB^3.$ But the contrary is also true by \cite[Prop. VII.5.7]{Simon}, namely that for any $O \in \mathbf{SO}(3)$ there exists a $U \in \sut$ such that the action $\rho \mapsto U \rho U^*$ is described by $O$ in the Bloch sphere model.
\par
Moreover, as $\rho \mapsto \overline{\rho}$ (taking the element-wise complex conjugate in the standard basis) is an orientation reversing isometry in the Bloch sphere model --- namely, it is the reflection to the \lq\lq $x-z$ plane" as $\overline{\sigma_1}=\sigma_1, \, \overline{\sigma_3}=\sigma_3,$ and $\overline{\sigma_2}=-\sigma_2$ ---, for every $O \in \mathbf{O}(3)$ there is a unitary or anti-unitary $U$ such that $\rho \mapsto U \rho U^*$ is described by $O$ in the Bloch sphere model. Unitary conjugations correspond to orientation-preserving isometries of $\bB^3,$ while anti-unitary conjugations correspond to orientation-reversing isometries. 
\par
So \eqref{eq:isom-full-O3} tells us that all isometries \emph{must be} unitary or anti-unitary conjugations.
The converse statement has been proved in Lemma \ref{lem:wig-sym} and hence the proof is done.
\end{proof}

\section{Non-injective and non-surjective isometries}

We turn to the case when the cost operator involves the qubit \lq\lq clock" and \lq\lq shift" operators which are intimately related to the finite dimensional approximations of the position and momentum operators in quantum mechanics \cite{satan-tequila, singh-carroll}.
\par
As the qubit \lq\lq clock" operator is $\sigma_3$ and the \lq\lq shift" is $\sigma_1,$ let us define the corresponding cost operator $C_{xz}$ by
\be \label{eq:real-cost-def}
C_{xz}:=\sum_{j=1,3} \ler{\sigma_j \otimes I_{\cH^*} - I_{\cH}  \otimes \sigma_j^T}^2=4 I -2 \sum_{j=1,3} \sigma_j \otimes \sigma_j^T
=\lesq{\ba{cccc} 2 & 0 & 0 & -2 \\ 0 & 6 & -2 & 0 \\ 0 & -2 & 6 & 0 \\ -2 & 0 & 0 & 2 \ea}.
\ee
We denote the corresponding quantum Wasserstein distance by $D_{xz}.$ We need some additional definitions to state the main result of this section precisely.
\par
Let $\cG$ denote the group of transformations of $\cS\ler{\cH}$ generated by the maps $\rho \mapsto e^{it\sigma_2} \rho e^{-it\sigma_2} \, \ler{t \in \R}$ and $\rho \mapsto \sigma_3 \overline{\rho} \sigma_3^{*}.$ That is,
\be \label{eq:G-def}
\cG:=\left<\lecurl{\rho \mapsto e^{it\sigma_2} \rho e^{-it\sigma_2} \, \middle| \, t \in \R} \cup \lecurl{\rho \mapsto \sigma_3 \overline{\rho} \sigma_3^{*}}\right>.
\ee
Note that in the Bloch sphere model these transformations are rotations around the \lq\lq $y$ axis" and the reflection to the \lq\lq $y-z$ plane" and hence $\cG \cong \bO(2).$
\par
Let $\cK$ denote the group of self-maps of $\cS\ler{\cH}$ generated by taking the complex conjugate in the standard basis, that is, by the map $\rho \mapsto \overline{\rho}.$ Note that every element of $\cG$ commutes with every element of $\cK,$ and $\cK$ is isomorphic to $\cC_2,$ the cyclic group of order $2.$
\par
Let $\cS^\R\ler{\cH}$ denote the set of all real symmetric states
$$
\cS^\R\ler{\cH}=\lecurl{\fel\ler{I_{\cH}+x \sigma_1+z \sigma_3}\, \middle| \, x^2+z^2 \leq 1},
$$
and let $\cP_1^\R\ler{\cH}$ denote the set of all real symmetric pure states, that is, $$\cP_1^\R\ler{\cH}=\cP_1\ler{\cH} \cap \cS^\R\ler{\cH}.$$
\par
For a set $\cD \subseteq \cS\ler{\cH},$ we define
$$
\cF_{\lecurl{-1,1}}^{\ler{\cD}}:=\lecurl{\xi: \cS\ler{\cH}\to{ \cS\ler{\cH}} \, \middle| \, \xi\ler{\rho} \in \lecurl{\rho, \overline{\rho}} \text{ if } \rho \in \cD \text{ and } \xi\ler{\rho}=\rho \text{ if } \rho\notin\cD}.
$$
We will consider the special cases where $\cD=\cP_1\ler{\cH} \setminus \cP_1^\R\ler{\cH}$ or $\cD=\cS\ler{\cH} \setminus \cS^\R\ler{\cH}.$
Note that in these special cases $\cF_{\lecurl{-1,1}}^{\ler{\cD}}$ endowed with the composition operation is a unital semigroup. Moreover, the restriction of $\psi$ to $\cD$ is a bijection of $\cD$ for every $\psi \in \cG \times \cK.$
\par
Let us take a closer look at the algebraic structure of transformations of $\cS\ler{\cH}$ generated by the elements of $\cG \times \cK$ and $\cF_{\lecurl{-1,1}}^{\ler{\cP_1\ler{\cH} \setminus \cP_1^\R\ler{\cH}}}.$ For every $\psi \in \cG \times \cK,$ the map
$$
\xi \mapsto \psi^{-1} \circ \xi \circ \psi \qquad \ler{\xi \in \cF_{\lecurl{-1,1}}^{\ler{\cP_1\ler{\cH} \setminus \cP_1^\R\ler{\cH}}}}
$$
is an automorphism of $\cF_{\lecurl{-1,1}}^{\ler{\cP_1\ler{\cH} \setminus \cP_1^\R\ler{\cH}}}.$
Therefore, any composition $\Phi$ of elements of $\cG \times \cK$ and $\cF_{\lecurl{-1,1}}^{\ler{\cP_1\ler{\cH} \setminus \cP_1^\R\ler{\cH}}}$ can be written in a simple form: $\Phi=\psi \circ \xi$ for some $\psi \in \cG \times \cK$ and $\xi \in \cF_{\lecurl{-1,1}}^{\ler{\cP_1\ler{\cH} \setminus \cP_1^\R\ler{\cH}}}.$
\par
So the algebraic structure of transformations of $\cS\ler{\cH}$ generated by the elements of $\cG \times \cK$ and $\cF_{\lecurl{-1,1}}^{\ler{\cP_1\ler{\cH} \setminus \cP_1^\R\ler{\cH}}}$ is the \emph{semidirect product}
$$
\ler{\cG \times \cK} \ltimes_{\varphi_1} \cF_{\lecurl{-1,1}}^{\ler{\cP_1\ler{\cH} \setminus \cP_1^\R\ler{\cH}}} 
$$
where $\ltimes_{\varphi_1}$ is defined by the standard action
\be \label{eq:std-1}
\ler{\varphi_1(\psi)}(\xi)=\psi^{-1} \circ \xi \circ \psi.
\ee
That is, the action of an element $\ler{\psi, \xi} \in \ler{\cG \times \cK} \ltimes_{\varphi_1} \cF_{\lecurl{-1,1}}^{\ler{\cP_1\ler{\cH} \setminus \cP_1^\R\ler{\cH}}}$ on the state space $\cS\ler{\cH}$ is given by $\ler{\psi, \xi}\ler{\rho}=\ler{\psi \circ \xi}\ler{\rho},$ and the product of $\ler{\psi', \xi'}$ and $\ler{\psi, \xi}$ is
$$
\ler{\psi', \xi'}*\ler{\psi, \xi}=\ler{\psi' \circ \psi, \,   \ler{\varphi_1(\psi)}\ler{\xi'} \circ \xi}=\ler{\psi' \circ \psi, \,   \psi^{-1} \circ \xi' \circ \psi \circ \xi},
$$
and hence for $\rho \in \cS\ler{\cH},$
$$
\ler{\ler{\psi', \xi'}*\ler{\psi, \xi}}\ler{\rho}= \ler{\psi' \circ \psi \circ  \psi^{-1} \circ \xi' \circ \psi \circ \xi} \ler{\rho}=\ler{\psi' \circ \xi' \circ \psi \circ \xi} \ler{\rho}.
$$
Very similarly, the elements of $\cG$ give rise to automorphisms of $\cF_{\lecurl{-1,1}}^{\ler{\cS\ler{\cH} \setminus \cS^\R\ler{\cH}}}$ by conjugation, and the algebraic structure generated by the elements of $\cG$ and $\cF_{\lecurl{-1,1}}^{\ler{\cS\ler{\cH} \setminus \cS^\R\ler{\cH}}}$ is the semidirect product
$$
\cG  \ltimes_{\varphi_2} \cF_{\lecurl{-1,1}}^{\ler{\cS\ler{\cH} \setminus \cS^\R\ler{\cH}}} 
$$
where $\ltimes_{\varphi_2}$ is defined by the standard action
\be \label{eq:std-2}
\ler{\varphi_2(\psi)}(\xi)=\psi^{-1} \circ \xi \circ \psi.
\ee

Let $\isom{\cW_2^{(xz)}\ler{\cS\ler{\cH}}}$ denote the semigroup of quantum Wasserstein isometries of the state space $\cS\ler{\cH}$ with respect to the cost operator $C_{xz}$ given in \eqref{eq:real-cost-def}.

\begin{theorem} \label{thm:clock-shift}
If $\psi \in \cG \times \cK \cong \bO(2) \times \cC_2$ and $\xi \in \cF_{\lecurl{-1,1}}^{\ler{\cP_1\ler{\cH} \setminus \cP_1^\R\ler{\cH}}},$ then the map $\Phi=\psi \circ \xi$ belongs to the semigroup $\isom{\cW_2^{(xz)}\ler{\cS\ler{\cH}}}.$
On the other hand, if $\Phi \in \isom{\cW_2^{(xz)}\ler{\cS\ler{\cH}}},$ then there exists a unique $\psi \in \cG  \cong \bO(2)$ and a unique $\xi \in \cF_{\lecurl{-1,1}}^{\ler{\cS\ler{\cH} \setminus \cS^\R\ler{\cH}}}$ such that $\Phi=\psi \circ \xi.$
In other words,
\be \label{eq:bounds}
  \ler{\bO(2) \times \cC_2} \ltimes_{\varphi_1} \cF_{\lecurl{-1,1}}^{\ler{\cP_1\ler{\cH} \setminus \cP_1^\R\ler{\cH}}} \subseteq \isom{\cW_2^{(xz)}\ler{\cS\ler{\cH}}} \subseteq
 \bO(2) \ltimes_{\varphi_2}  \cF_{\lecurl{-1,1}}^{\ler{\cS\ler{\cH} \setminus \cS^\R\ler{\cH}}}
\ee
where the semidirect products $\ltimes_{\varphi_1}$ and $\ltimes_{\varphi_2}$ are defined by the standard actions
\eqref{eq:std-1} and \eqref{eq:std-2}, respectively.
\end{theorem}

\begin{proof}

Let us start with proving the lower bound in \eqref{eq:bounds}.
First we show that the orthogonal group $\bO(2)$, identified with $\cG$ defined in \eqref{eq:G-def}, is a subgroup of the Wasserstein isometry semigroup. That is, with $U(t):=e^{it\sigma_2} \, \ler{t \in \R},$ the unitary similarity transformation
\be \label{eq:unit-sim-real}
\rho \mapsto U(t) \rho U(t)^{*} \qquad \ler{\rho \in \cS\ler{\cH}}
\ee
is a quantum Wasserstein isometry for every $t \in \R,$ and the anti-unitary similarity transformation
\be \label{eq:anti-unit-sim-real}
\rho \mapsto \sigma_3 \overline{\rho} \sigma_3^{*} \qquad \ler{\rho \in \cS\ler{\cH}}
\ee
is also a quantum Wasserstein isometry.
\par
Let $\cU(t):=U(t) \otimes \ler{U(t)^T}^*.$ Note that $U(t)=\lesq{\ba{cc}\cos{t} & \sin{t}\\ -\sin{t} & \cos{t} \ea}.$ We show that the \lq\lq clock and shift" cost defined in \eqref{eq:real-cost-def} is invariant under conjugation by $\cU(t),$ that is,
\be \label{eq:cost-invariance-real}
\cU(t) \, C_{xz} \, \cU(t)^*=C_{xz}
\ee
for all $t \in \R.$
Indeed, 
\begin{equation*}
U(t) \sigma_3 U(t)^*=\lesq{\ba{cc}\cos{2t} & -\sin{2t}\\ -\sin{2t} & -\cos{2t} \ea}~~~\mbox{ and }~~~ U(t) \sigma_1 U(t)^*=\lesq{\ba{cc}\sin{2t} & \cos{2t}\\ \cos{2t} & -\sin{2t} \ea},
\end{equation*}
and hence
$$
\cU(t) \, C_{xz} \, \cU(t)^*
=\cU(t) \ler{4 I_{\hohc} -2 \sum_{j=1,3} \sigma_j \otimes \sigma_j^T} \cU(t)^*
$$
$$
=4 I_{\hohc} -2 \sum_{j=1,3} U(t)\sigma_j U(t)^* \otimes \ler{U(t)^T}^*\sigma_j^T U(t)^T 
$$
$$
=4 I_{\hohc}- 
2 \left(\lesq{\ba{cccc}\cos^2{2t} & -\cos{2t}\sin{2t} & -\cos{2t}\sin{2t} & \sin^2{2t}\\
-\cos{2t}\sin{2t} & -\cos^2{2t} & \sin^2{2t} &\cos{2t}\sin{2t}\\
-\cos{2t}\sin{2t} & \sin^2{2t} & -\cos^2{2t} & \cos{2t}\sin{2t} \\
\sin^2{2t} & \cos{2t}\sin{2t} & \cos{2t}\sin{2t} & \cos^2{2t}
\ea}+ \right.
$$
$$
\left.
+
\lesq{\ba{cccc}\sin^2{2t} & \cos{2t}\sin{2t} & \cos{2t}\sin{2t} & \cos^2{2t} \\
\cos{2t}\sin{2t} & -\sin^2{2t} & \cos^2{2t} & -\cos{2t}\sin{2t}\\
\cos{2t}\sin{2t} & \cos^2{2t} & -\sin^2{2t} & -\cos{2t}\sin{2t} \\
\cos^2{2t} & -\cos{2t}\sin{2t} & -\cos{2t}\sin{2t} & \sin^2{2t}
\ea} \right)
$$
$$
=
\lesq{\ba{cccc} 4-2 & 0 & 0 & -2 \\
0 & 4+2 & -2 & 0\\
0 & -2 & 4+2 & 0 \\
-2 & 0 & 0 & 4-2
\ea}=C_{xz}.
$$

Now, using \eqref{eq:coup-unit-inv}, an argument very similar to \eqref{eq:d-sym-inv} shows that this unitary invariance of the cost \eqref{eq:cost-invariance-real} implies
$$
D_{xz}\ler{U(t) \rho U(t)^*, U(t) \omega U(t)^*}=D_{xz}\ler{\rho, \omega} \qquad \ler{\rho, \omega \in \cS\ler{\cH}}.
$$
The next step is considering the anti-unitary similarity transformation \eqref{eq:anti-unit-sim-real}.
Note that if $\Pi \in \cC \ler{\rho, \omega}$ and $\Pi=\sum_{r=1}^R A_r \otimes B_r^T \, \ler{A_r,B_r \in \cL\ler{\cH}},$ then
$$
\ler{\sigma_3 \otimes \ler{\sigma_3^T}^*} \, \overline{\sum_{r=1}^R A_r \otimes B_r^T} \,\ler{\sigma_3^* \otimes \sigma_3^T}
=\sum_{r=1}^R \sigma_3 \overline{A_r} \sigma_3^* \otimes \ler{\sigma_3 \overline{B_r} \sigma_3^*}^T
$$
and hence 
\be \label{eq:au-tr}
\lecurl{\ler{\sigma_3 \otimes \ler{\sigma_3^T}^*} \, \overline{\Pi} \,\ler{\sigma_3^* \otimes \sigma_3^T} \, \middle| \, \Pi \in \cC \ler{\rho, \omega} }
\subseteq
\cC\ler{\sigma_3 \overline{\rho} \sigma_3^{*}, \sigma_3 \overline{\omega} \sigma_3^{*}}.
\ee
However, the inclusion in the other direction can be shown very similarly: if $\Sigma \in \cC\ler{\sigma_3 \overline{\rho} \sigma_3^{*}, \sigma_3 \overline{\omega} \sigma_3^{*}},$ then 
$$
\ler{\sigma_3^* \otimes \sigma_3^T} \, \overline{\Sigma} \,\ler{\sigma_3 \otimes \ler{\sigma_3^T}^*} \in \cC\ler{\rho, \omega},
$$
and hence \eqref{eq:au-tr} holds with equality. With this in hand,
$$
D_{xz}^2\ler{\sigma_3 \overline{\rho} \sigma_3^*, \sigma_3 \overline{\omega} \sigma_3^*}
=\inf \lecurl{\tr_{\hohc} \ler{C_{xz} \Sigma} \, \middle| \, \Sigma \in \cC\ler{\sigma_3 \overline{\rho} \sigma_3^{*}, \sigma_3 \overline{\omega} \sigma_3^{*}}}
$$
$$
=\inf \lecurl{\tr_{\hohc} \ler{C_{xz} \ler{\sigma_3 \otimes \ler{\sigma_3^T}^*} \, \overline{\Pi} \,\ler{\sigma_3^* \otimes \sigma_3^T}} \, \middle| \, \Pi \in \cC\ler{\rho, \omega}}
$$
$$
=\inf \lecurl{\tr_{\hohc} \ler{ \ler{\sigma_3^* \otimes \sigma_3^T} \, C_{xz} \ler{\sigma_3 \otimes \ler{\sigma_3^T}^*} \, \overline{\Pi} } \, \middle| \, \Pi \in \cC\ler{\rho, \omega}}
$$

\be \label{eq:d-xz-inv}
=\inf \lecurl{\tr_{\hohc} \ler{C_{xz} \, \Pi } \, \middle| \, \Pi \in \cC\ler{\rho, \omega}}=D_{xz}^2\ler{\rho,  \omega},
\ee
where we used that cost \eqref{eq:real-cost-def} is invariant under the transformation \eqref{eq:anti-unit-sim-real} in the sense that
$$
\overline{\ler{\sigma_3^* \otimes \sigma_3^T} \, C_{xz} \ler{\sigma_3 \otimes \ler{\sigma_3^T}^*}}
$$
$$
=\sum_{j \in \{1,3\}}\ler{\sigma_3 \overline{\sigma_j} \sigma_3^{*} \otimes \sigma_3 \overline{I_{\cH^*}} \sigma_3^{*} - \sigma_3 \overline{I_{\cH}} \sigma_3^{*} \otimes \sigma_3 \overline{\sigma_j} \sigma_3^{*}}^2
$$
$$
=\sum_{j=1,3} \ler{\sigma_j \otimes I_{\cH^*} - I_{\cH}  \otimes \sigma_j^T}^2=C_{xz}.
$$
So this anti-unitary similarity transformation \eqref{eq:anti-unit-sim-real} is also a Wasserstein isometry.
\par
To see that the map $\rho \mapsto \overline{\rho}$ is also isometric, we only have to note that
$$
D_{xz}^2\ler{\overline{\rho}, \overline{\omega}}
=\inf \lecurl{\tr_{\hohc} \ler{C_{xz} \Sigma} \, \middle| \, \Sigma \in \cC\ler{\overline{\rho}, \overline{\omega}}}
$$
$$
=\inf \lecurl{\tr_{\hohc} \ler{C_{xz} \overline{\Pi}} \, \middle| \, \Pi \in \cC\ler{\rho, \omega}}
=\inf \lecurl{\tr_{\hohc} \ler{\overline{C_{xz}} \Pi} \, \middle| \, \Pi \in \cC\ler{\rho, \omega}}
$$
$$
=\inf \lecurl{\tr_{\hohc} \ler{C_{xz} \Pi} \, \middle| \, \Pi \in \cC\ler{\rho, \omega}}
=D_{xz}^2\ler{\rho, \omega}.
$$
The last step in proving the lower bound in \eqref{eq:bounds} is showing the isometric property of maps of the form 
\begin{align} \label{eq:weird-map}
    \xi\ler{\rho}:= \left\{ \begin{matrix}
    \rho & \rho \in \cS\ler{\cH} \setminus \cP_1\ler{\cH} \\
    \fel\ler{\rho+\overline{\rho}}+\epsilon_2\ler{\rho} \fel\ler{\rho-\overline{\rho}}  & \rho \in \cP_1\ler{\cH}
    \end{matrix}\right.
\end{align} 
where $\epsilon_2\colon \cP_1\ler{\cH}\to \{-1,1\}$ is an \emph{arbitrary} function.
\par
Let $\rho$ and $\omega$ be defined as in \eqref{eq:rho-def} and \eqref{eq:omega-def}. The cost of the trivial coupling of $\rho$ and $\omega$ is
$$
\tr_{\hohc} \ler{\ler{\omega \otimes \rho^T} C_{xz}}=\tr_{\hohc} \ler{\ler{\omega \otimes \rho^T} \ler{4 I_{\hohc} -2 \sum_{j=1,3} \sigma_j \otimes \sigma_j^T}}
$$
\be \label{eq:triv-coupl-cost-real}
=4-2 \sum_{j=1,3} \tr_{\cH} \ler{\omega \sigma_j} \cdot \tr_{\cH^*} \ler{\rho^T \sigma_j^T} = 4- 2 \ler{xu+zw}.
\ee
Recall that $\cC\ler{\rho, \omega}=\lecurl{\omega \otimes \rho^T}$ if either $\rho$ or $\omega$ is a pure state. So the isometric property of \eqref{eq:weird-map} holds because by \eqref{eq:triv-coupl-cost-real}, for any state $\rho$ and any \emph{pure} state $\omega$ given as in \eqref{eq:rho-def} and \eqref{eq:omega-def} we have
\begin{align*}
  D_{xz}^2(\xi(\omega),\xi(\rho)) =  D_{xz}^2(\omega,\rho) = 4 - 2(xu+zw)
\end{align*}
where we used that $\xi$ defined in \eqref{eq:weird-map} does not change the first and third elements of the Bloch vectors of $\rho$ and $\omega.$ So if at least one of the states involved is pure, the distance does not depend on the second elements of their Bloch vectors at all.
\par
So any composition of the transformations in $\cG \times \cK \cong \bO(2) \times \cC_2$ and maps from $\cF_{\lecurl{-1,1}}^{\ler{\cP_1\ler{\cH} \setminus \cP_1^\R\ler{\cH}}},$ that is, maps of the form \eqref{eq:weird-map}, are Wasserstein isometries with respect to $C_{xz}.$
In other words, 
$$
\ler{\bO(2) \times \cC_2} \ltimes_{\varphi_1} \cF_{\lecurl{-1,1}}^{\ler{\cP_1\ler{\cH} \setminus \cP_1^\R\ler{\cH}}} \subseteq \isom{\cW_2^{(xz)}\ler{\cS\ler{\cH}}},
$$
so the lower bound in \eqref{eq:bounds} is proved.
\par
To prove the upper bound in \eqref{eq:bounds}, let us recall the transport costs of trivial couplings given in \eqref{eq:triv-coupl-cost-real}. An immediate consequence of \eqref{eq:triv-coupl-cost-real} is that
\be \label{eq:max-dist-real}
\max_{\rho, \omega \in \cS\ler{\cH}} D_{xz}^2\ler{\rho, \omega}=6
\ee
and the maximum is attained if and only if $\rho$ and $\omega$ are orthogonal pure states with real elements, that is,
$$
\rho=\fel\ler{I_{\cH}+\cos\alpha \sigma_1+\sin \alpha \sigma_3} \text{ and } \omega=\fel\ler{I_{\cH}-\cos\alpha \sigma_1-\sin \alpha \sigma_3}
$$
for some $\alpha \in \R.$
\par
Consequently, any quantum Wasserstein isometry with respect to $C_{xz}$ maps real symmetric pure states to real symmetric pure states.

\par
Another easy consequence of \eqref{eq:triv-coupl-cost-real} is that the quantum Wasserstein distance of real symmetric pure states $\rho, \omega \in \cP_1^{\R}\ler{\cH}$ can be expressed in terms of their Bloch vectors as follows:
\be \label{eq:pure-st-dist-real}
D_{xz}^2\ler{\rho, \omega}=4-2\inner{\bb_\rho}{\bb_\omega}=2+\norm{\bb_\rho-\bb_\omega}^2 \qquad \ler{\rho, \omega \in \cP_1^\R\ler{\cH}}.
\ee
Therefore, any map $\Phi: \, \cS\ler{\cH} \rightarrow \cS\ler{\cH}$ satisfying 
$$
D_{xz}\ler{\Phi\ler{\rho}, \Phi\ler{\omega}}=D_{xz}\ler{\rho, \omega} \qquad \ler{\rho, \omega \in \cS\ler{\cH}}
$$
acts on real symmetric pure states as an isometry of the circle $\bS^1$ representing $\cP_1^{\R}\ler{\cH}$ in the Bloch sphere model. That is, there exists an $O \in \bO(2)$ such that \be \label{eq:isom-O2}
\bb_{\Phi\ler{\rho}}=O\ler{\bb_\rho} \qquad \ler{\rho \in \cP_1^{\R}\ler{\cH}}.
\ee
However, we have seen before that the elements of $\bO(2)$ identified with the elements of $\cG$ are Wasserstein isometries. Therefore, for any isometry $\Phi$ there exists a unique $\psi \in \cG$ such that $\xi:=\psi^{-1} \circ \Phi$ is a Wasserstein isometry that leaves every real symmetric pure state fixed. In particular, $\omega_1=\fel\ler{I_{\cH}+\sigma_1}$ and $\omega_3=\fel\ler{I_{\cH}+\sigma_3}$ are fixed by $\xi,$ and hence \eqref{eq:triv-coupl-cost-real} tells us that for $\rho \in \cS\ler{\cH}$ given as in \eqref{eq:rho-def} we have
\be \label{eq:}
D^2_{xz}\ler{\xi\ler{\rho},\omega_1}=4-2x=D^2_{xz}\ler{\rho,\omega_1} \text{ and } D^2_{xz}\ler{\xi\ler{\rho},\omega_3}=4-2z=D^2_{xz}\ler{\rho,\omega_3}.
\ee
Consequently, $\xi$ cannot change the \lq\lq x" and \lq\lq z" coordinate of the Bloch vector of a state $\rho.$ In other words, $\xi$ leaves the line segment
\be
\ell_{(x,z)} := \lecurl{\fel\ler{I_{\cH} + x \sigma_1 + y \sigma_2 + z \sigma_3}\colon y^2\leq 1-x^2-z^2}
\ee
invariant, that is, $\xi\ler{\ell_{(x,z)}} \subseteq \ell_{(x,z)}.$
\par
Moreover, the \lq\lq y" coordinate of the Bloch vector of a state is also quite fixed: it is either left invariant, or mapped to its negative.
\par
To show this, we shall utilise the result of De Palma and Trevisan which is an explicit formula for the self-distances of states \cite[Corollary 1]{DPT}. It tells us that for any quadratic cost operator $C$ we have
\be \label{eq:self-dist}
D_C^2\ler{\rho, \rho}=\tr_{\hohc} \ler{C \, \Ket{\sqrt{\rho}} \Bra{\sqrt{\rho}}}
\ee
where $\Ket{\sqrt{\rho}} \Bra{\sqrt{\rho}} \in \cS\ler{\cH \otimes \cH^*}$ is the canonical purification of the state $\rho \in \cS\ler{\cH}$ -- see \cite{Holevo}.
\par
 Notice that the eigenvalues of $C_{xz}$ are $8,4,4,0$, and that its spectral decomposition is
$$
C_{xz} = 2\Ket{\sigma_1}\Bra{\sigma_1} +4\Ket{\sigma_2}\Bra{\sigma_2} + 2\Ket{\sigma_3}\Bra{\sigma_3}.
$$
Consider an arbitrary non-tracial state $\rho$ written in its spectral decomposition
\be \label{eq:rho-spectral}
\rho = \fel\ler{I_{\cH}+ \bb_\rho\cdot\Sig} = \lambda\cdot\fel\ler{I_{\cH} + \frac{\bb_\rho}{\norm{\bb_\rho}}\cdot\Sig} + (1-\lambda)\cdot\fel\ler{I_{\cH} - \frac{\bb_\rho}{\norm{\bb_\rho}}\cdot\Sig}
\ee
where $\Sig$ is the vector containing the Pauli matrices, $\norm{\bb_\rho} = (2\lambda-1), \, \lambda\in\left(\fel,1\right]$ and hence $\lambda = \fel\ler{1+\norm{\bb_\rho}}.$
Therefore
\begin{equation}\label{eq:sqrtrho-spectral}
\sqrt{\rho} = \sqrt{\lambda}\cdot\fel\ler{I_{\cH} + \frac{\bb_\rho}{\norm{\bb_\rho}}\cdot\Sig} + \sqrt{1-\lambda}\cdot\fel\ler{I_{\cH} - \frac{\bb_\rho}{\norm{\bb_\rho}}\cdot\Sig},
\end{equation}
and
\begin{align*}
    &D_{xz}^2(\rho,\rho) = \langle\langle\sqrt{\rho}||C_{xz}||\sqrt{\rho}\rangle\rangle \\
    &= \left\langle\left\langle \frac{\sqrt{\lambda}+\sqrt{1-\lambda}}{2}I_{\cH} + \frac{\sqrt{\lambda}-\sqrt{1-\lambda}}{2}(x,y,z)\cdot\Sig \right|\left| \frac{\sqrt{\lambda}-\sqrt{1-\lambda}}{2}(4x,8y,4z)\cdot\Sig \right\rangle\right\rangle \\
    &= \ler{\sqrt{\lambda}-\sqrt{1-\lambda}}^2\left\langle\left\langle (x,y,z)\cdot\Sig \right|\left| (x,2y,z)\cdot\Sig \right\rangle\right\rangle
    = \ler{1-2\sqrt{\lambda}\sqrt{1-\lambda}} \ler{2 + 2 y^2} \\
    &=2\ler{1-\sqrt{1-\norm{\bb_\rho}^2}} \ler{1 + \frac{y^2}{\norm{\bb_\rho}^2}}.
\end{align*}
For fixed $x$ and $z,$ this expression is strictly monotone increasing in $y^2$, and hence the preserver equation $D_{xz}^2\ler{\xi\ler{\rho},\xi\ler{\rho}}=D_{xz}^2\ler{\rho,\rho}$ tells us that the second element of the Bloch vector of $\xi\ler{\rho}$ is either $y$ or $-y.$
This means that $\xi \in \cF_{\lecurl{-1,1}}^{\ler{\cS\ler{\cH} \setminus \cS^\R\ler{\cH}}},$ and hence 
$$
\Phi=\psi \circ \xi \in \bO(2) \ltimes_{\varphi_2}  \cF_{\lecurl{-1,1}}^{\ler{\cS\ler{\cH} \setminus \cS^\R\ler{\cH}}}
$$
as desired.
\end{proof}

\section{Numerics on the clock and shift case}\label{s: numerics}

We performed some numerical test using Wolfram Mathematica \cite{WM} to study Theorem \ref{thm:clock-shift} where we obtained a lower and an upper bound for the isometry semigroup, see \eqref{eq:bounds}. The Mathematica notebook along with its pdf image is available online, see \cite{VW}.
\par
These numerical tests suggest that the truth in Theorem \ref{thm:clock-shift} is the \emph{lower bound.} That is,
\be \label{eq:num-truth}
\isom{\cW_2^{(xz)}\ler{\cS\ler{\cH}}}=\ler{\bO(2) \times \cC_2} \ltimes_{\varphi_1} \cF_{\lecurl{-1,1}}^{\ler{\cP_1\ler{\cH} \setminus \cP_1^\R\ler{\cH}}}.
\ee
The meaning of \eqref{eq:num-truth} is that Wasserstein isometries fixing real symmetric pure states behave \emph{uniformly} on mixed states: they either send all mixed states to their conjugate, or they leave all mixed states fixed. This is in striking contrast with the behaviour of isometries on pure states, where the conjugate can be taken or omitted independently, see \eqref{eq:weird-map}.
\par
The argument supporting this uniform property of the isometries reads as follows. Assume that an isometry $\xi$ leaves every element of $\cP_1^\R\ler{\cH}$ and hence $\cS^\R\ler{\cH}$ fixed. That is, we factor out by $\cG \cong \bO(2).$ Now, for every state $\rho \in \cS\ler{\cH}$ we have either $\xi\ler{\rho}=\rho$ or $\xi\ler{\rho}=\overline{\rho}.$
\par
We choose the distinguished mixed state
\be \label{eq:dist-mixed}
\eta:=\fel\ler{I_{\cH}+\fel \sigma_2}
\ee
and show numerical evidences that suggest that if $\xi$ sends $\eta$ to its conjugate, then it sends all \emph{mixed states} to their conjugates. And vice versa: if $\eta$ is left invariant then so are all the mixed states.
Note that there seems to be a crucial difference between mixed and pure states: as \eqref{eq:weird-map} shows, the \emph{pure states} may be left invariant or sent to their conjugates completely independently of the action of $\xi$ on other states.
\par
The following numerical results suggest that
\be \label{eq:num-ineq}
D_{xz}\ler{\rho, \eta}< D_{xz}\ler{\overline{\rho}, \eta}
\ee
for any $\rho \in \cS\ler{\cH} \setminus \cP_1\ler{\cH}$ with $\lesq{\bb_\rho}_2>0,$ that is, with positive \lq\lq y" coordinate in the Bloch model.
\par
In the experiment shown in Figure \ref{fig:fel} we fix the $\sigma_2$ component of
$$
\rho=\fel\ler{I_{\cH}+x \sigma_1+y \sigma_2+z \sigma_3}
$$
as $y=1/2$ (yellow) and $y=-1/2$ (blue), and we let both $x$ and $z$ run from $-\sqrt{3/8}$ to $\sqrt{3/8}.$
\begin{figure}[H]
    \centering
    \includegraphics[width=3.1in]{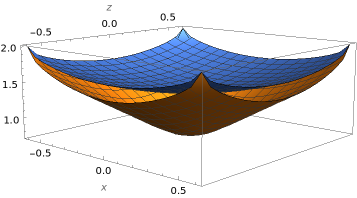}
    \caption{$y=\pm 1/2, \, -\sqrt{3/8}\leq x,z \leq \sqrt{3/8}$}
    \label{fig:fel}
\end{figure}
It is clear from the picture that
\be \label{eq:num1}
D_{xz} \ler{\fel\ler{I_{\cH}+x \sigma_1+\fel \sigma_2+z \sigma_3},\eta}
<
D_{xz} \ler{\fel\ler{I_{\cH}+x \sigma_1-\fel \sigma_2+z \sigma_3},\eta}
\ee
if $x^2+z^2<3/4$ (mixed $\rho$) and the two sides of \eqref{eq:num1} coincide for $x^2+z^2=3/4$ (pure $\rho$).
\par
In the second experiment we increased $\abs{y}$ to $4/5.$ In this case the difference between $D_{xz}\ler{\rho, \eta}$ and $D_{xz}\ler{\overline{\rho}, \eta}$ is even more visible, see Figure \ref{fig:negyotod}.
\begin{figure}[H]
    \centering
    \includegraphics[width=3.6in]{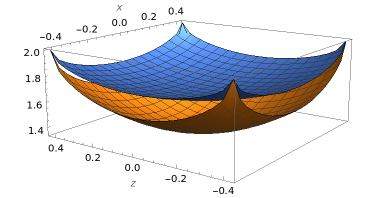}
    \caption{$y=\pm 4/5, \, -\sqrt{9/50}\leq x,z \leq \sqrt{9/50}$}
    \label{fig:negyotod}
\end{figure}

A natural guess is that if $\abs{y}$ is small, then so is the difference between the two sides of \eqref{eq:num-ineq}. This is indeed the case: the difference is so small that it cannot be seen well on $2$-dimensional plot. But a $1$-dimensional section that we obtain by letting $x=z$ shows again that we have strict inequality for mixed states and equality for pure states --- see Figure \ref{fig:egykilenced}.

\begin{figure}[H]
    \centering
    \includegraphics[width=3.2in]{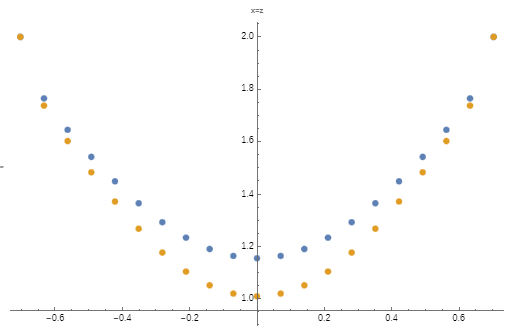}
    \caption{$y=\pm 1/9, \, -\sqrt{40/81}\leq x=z \leq \sqrt{40/81}$}
    \label{fig:egykilenced}
\end{figure}

\paragraph*{{\bf Acknowledgement}}
We are grateful to the anonymous reviewer for his/her valuable comments and insightful suggestions.

\end{document}